\documentclass[11pt]{article}
\usepackage{geometry}
\usepackage{amsmath, amssymb, amsthm}
\usepackage{graphicx}
\usepackage{bm}
\usepackage[linesnumbered,ruled]{algorithm2e}
\usepackage{fullpage}
\usepackage{mathtools}
\usepackage{url}
\newtheorem{theorem}{Theorem}
\newtheorem{lemma}{Lemma}

\newtheorem{claim}{Claim}
\newtheorem{definition}{Definition}
\begin{document}

\title{The Strongly Stable Roommates Problem and Linear Programming}

\author{Naoyuki Kamiyama\thanks{%
This work was supported by JST ERATO Grant Number JPMJER2301, Japan.}}
\date{Institute of Mathematics for Industry \\ 
Kyushu University \\ 
Fukuoka, Japan \\
\url{kamiyama@imi.kyushu-u.ac.jp}}

\maketitle

\begin{abstract}
The stable roommates problem is a non-bipartite version 
of the stable matching problem in a bipartite graph. 
In this paper, we consider the stable roommates problem
with ties. 
In particular, we focus on 
strong stability, which is one of the main 
stability 
concepts in the stable roommates problem
with ties. 
We propose a new polynomial-time algorithm for the 
problem 
of checking the existence of a strongly stable matching 
in the stable roommates problem with ties.
More concretely, 
we extend 
the linear programming approach of Abeledo and Blum
to the stable roommates problem with strict preferences
to our problem. 
\end{abstract} 

%{\bf Keywords:}
%stable matching, the stable roommates problem, strong stability

\section{Introduction} 

The topic of this paper is 
the stable roommates problem, which 
is a non-bipartite version 
of the well-known stable matching problem in a bipartite 
graph~\cite{GaleS62}. 
In the stable matching problem, there always exists a stable matching~\cite{GaleS62}. 
By contrast, it is known that 
a stable matching may not exist in the stable roommates problem~\cite{GaleS62}. 
Thus, the problem of checking the existence of 
a stable matching is one of the most important problems in the study of the stable
roommates problem.
For this problem, 
Irving~\cite{Irving85} proposed a 
polynomial-time algorithm.
This algorithm is combinatorial, i.e., 
this does not need solving linear programs.

In addition to combinatorial algorithms, 
linear programming approaches have been actively 
studied in the study of the stable matching problem
(see, e.g., \cite{Fleiner03,Rothblum92,VandeVate89}). 
Thus, it is natural to investigate a linear programming 
approach to the stable roommates problem. 
In this direction, 
Abeledo and Blum~\cite{AbeledoB96}
proved that 
we can check the existence of a stable matching in the stable roommates 
problem in polynomial time
by solving a linear program several times 
(see \cite{AbeledoR94} for related topics). 
Furthermore, 
Teo and Sethuraman~\cite{TeoS98}
proved that, 
for each instance of the stable roommates problem
in a complete graph, 
there exists a linear inequality system
such that there exists a feasible solution to 
this system if and only if 
there exists a stable matching in the given instance
(see \cite{SethuramanT01,TeoS00} for related work of 
\cite{TeoS98}).

In this paper, we consider the stable roommates problem
with ties, which is  
a variant of the stable roommates problem where 
agents may be indifferent between potential partners. 
In the study of the stable matching problem, 
the stable matching problem with ties 
has been actively studied 
(see, e.g., 
\cite{ChenG10,Irving94,IrvingMS00,IrvingMS03,Kamiyama22,
Kamiyama25+,KavithaMMP07,Kunysz19,Manlove99}). 
In the setting where ties are allowed in the preferences, 
weak stability, super-stability, and strong stability 
have been mainly considered. 
Weak stability 
guarantees that 
there does not exist an unmatched pair of agents such that 
both agents prefer the other agent in the pair to the current partner.
Super-stability 
guarantees
that there does not exist an unmatched pair of agents  
such that 
both agents weakly prefer the other agent in the pair to the current 
partner (i.e., each agent in the pair  
prefers the other agent in the pair to the current 
partner, or is indifferent between them). 
Strong stability 
guarantees that 
there does not exist an unmatched pair of agents such that
(i) 
both agents weakly prefer the other agent in the pair to the current 
partner, and 
(ii) 
at least one of the agents 
prefers the other agent in the pair to the current partner.
Ronn~\cite{Ronn90}
proved that 
the problem of checking the existence of a weakly stable
matching in the stable roommates problem with ties is NP-complete
(see 
also \cite{IrvingM02}). 
By contrast, 
the problem of checking the existence of 
a super-stable matching in the stable roommates problem with ties 
can be solved in 
polynomial time~\cite{FleinerIM07,FleinerIM11,IrvingM02}. 
Kamiyama~\cite{Kamiyama25}
extended the approach of Teo and Sethuraman~\cite{TeoS98} to 
super-stability in the stable roommates problem with ties. 

The topic of this paper is 
a strongly stable matching in the stable roommates problem
with ties.
Scott~\cite{Scott05} proposed a combinatorial polynomial-time 
algorithm
for the problem 
of checking the existence of a strongly stable matching 
in the stable roommates problem with ties. 
(Notice that the algorithm of \cite{Scott05} contains some flaws 
that can be removed
by using results in \cite{Kunysz16}.)
In addition, 
Kunysz~\cite{Kunysz16}
proposed a faster 
combinatorial polynomial-time 
algorithm for the same problem. 
This algorithm relies on the characterization of strongly
stable matchings in the stable roommates problem with ties 
in \cite{KunyszPG16}. 

The aim of this paper is to propose another
approach to the problem 
of checking the existence of a strongly stable matching 
in the stable roommates problem with ties.
More concretely, 
we propose a linear programming approach to 
the problem 
by extending the approach of Abeledo and Blum~\cite{AbeledoB96} to 
the stable roommates problem with strict preferences. 
Namely, we prove that 
we can solve
the problem of checking the existence of a strongly stable matching 
in the stable roommates problem with ties 
in polynomial time by solving a linear program 
several times.
Our extension is non-trivial 
in the following two points. 
\begin{itemize}
\item
The linear program used in \cite{AbeledoB96} 
is the same as the linear program used in 
\cite{Rothblum92,VandeVate89} 
to describe the convex hull of 
the characteristic vectors of stable matchings in the 
stable matching problem with strict preferences. 
Thus, it is natural to use the linear program 
used in \cite{Kunysz18}
to describe the convex hull of 
the characteristic vectors of strongly stable matchings 
in the stable matching problem with ties in a bipartite graph.
However, this does not work. More precisely, 
the key lemma (Lemma~\ref{lemma:existence}) 
of this paper does not hold. 
In order to resolve this issue, 
we add the constraints (the inequalities \eqref{eq_3:constraint}) used to 
describe the convex hull of 
the characteristic vectors of matchings in 
a general graph~\cite{Edmonds65}. 
\item
In the setting in \cite{AbeledoB96}, 
for each solution $x$ to the linear program and 
each vertex $v$, the best edge and the worst edge for $v$ 
in the support of $x$ are uniquely determined. 
By contrast, in our setting, since 
preferences contain ties, they may not be 
uniquely determined.
This makes the situation complicated. 
We resolve this issue by using 
the self-duality result for 
strongly stable matchings in a bipartite graph~\cite{Kunysz18}. 
\end{itemize}

The rest of this paper is organized as follows. 
In Section~\ref{section:preliminaries}, we give the formal 
definition of our setting. 
In Section~\ref{section:useful_lemmas}, we give 
lemmas that are needed in the next section. 
In Section~\ref{section:algorithm}, 
we propose our algorithm, and prove its correctness. 

\section{Preliminaries} 
\label{section:preliminaries} 

Let $\mathbb{R}_+$ denote the set of non-negative real numbers. 
For each finite set $U$, each vector $x \in \mathbb{R}_+^U$, 
and each subset $W \subseteq U$, we define 
$x(W) := \sum_{u \in W}x(u)$. 

In this paper, we are given a finite simple undirected graph $G = (V,E)$ with 
a vertex set $V$ and an edge set $E$.
In this paper, we do not distinguish between an edge $e$ and 
the set of end vertices of $e$. 
For each subset $F \subseteq E$ and 
each vertex $v \in V$, 
we define $F(v)$ 
as the set of edges $e \in F$ such that 
$v \in e$.
For each vertex $v \in V$, we are given a complete\footnote{%
For every pair of elements $e,f \in E(v) \cup \{\emptyset\}$, at least one of 
$e \succsim_v f$, $f \succsim_v e$ holds.} and transitive binary 
relation $\succsim_v$ 
on $E(v) \cup \{\emptyset\}$.
For each vertex $v \in V$ and 
each pair of elements $e,f \in E(v) \cup \{\emptyset\}$,
if $e \succsim_v f$ and $f \not\succsim_v e$
(resp.\ $e \succsim_v f$ and $f \succsim_v e$), then
we write $e \succ_v f$
(resp.\ $e \sim_v f$). 
Intuitively speaking, 
if $e \succ_v f$, then 
$v$ prefers $e$ to $f$. 
If $e \sim_v f$, then $v$ is indifferent between $e$ and $f$.
In this paper, we assume that, 
for every vertex $v \in V$ and every edge $e \in E(v)$, 
we have $e \succ_v \emptyset$. 

\begin{definition}
A subset $\mu \subseteq E$ is called a \emph{matching in $G$} if
$|\mu(v)| \le 1$ for every vertex $v \in V$.
\end{definition} 

For each matching $\mu$ in $G$ and each vertex $v \in V$ such that 
$\mu(v) \neq \emptyset$, 
we do not distinguish between $\mu(v)$ and the edge in $\mu(v)$. 

\begin{definition}
Let $\mu$ be a matching in $G$, and let $e$ be an edge in $E \setminus \mu$. 
\begin{itemize}
\item
For each vertex $v \in e$, 
we say that 
$e$ \emph{weakly blocks $\mu$ on $v$}
if 
$e \succsim_v \mu(v)$.
\item
For each vertex $v \in e$, 
we say that 
$e$ \emph{strongly blocks $\mu$ on $v$}
if 
$e \succ_v \mu(v)$.
\item
We say that 
$e$ \emph{blocks $\mu$}
if 
$e$ weakly blocks $\mu$ on every vertex in $e$, and 
$e$ strongly blocks $\mu$ on at least one of vertices in $e$.  
\end{itemize}
\end{definition} 

\begin{definition}
A matching $\mu$ in $G$ 
is said to be 
\emph{strongly stable}
if no edge in $E \setminus \mu$ 
blocks $\mu$. 
\end{definition}

For each matching $\mu$ in $G$, we define the 
characteristic 
vector 
$\chi_{\mu} \in \{0,1\}^E$ by 
$\chi_{\mu}(e) \coloneqq 1$ for each edge $e \in \mu$ and 
$\chi_{\mu}(e) \coloneqq 0$ for each edge $e \in E \setminus \mu$. 
For each subset $F \subseteq E$, a matching $\mu$ in $G$ is called 
a \emph{matching in $F$} if $\mu \subseteq F$. 
For each subset $X \subseteq V$ and 
each matching $\mu$ in $G$, we say that 
\emph{$\mu$ covers $X$} if 
$\mu(v) \neq \emptyset$ for every vertex $v \in X$. 

For each vector $x \in \mathbb{R}_+^E$, 
we define $E_x$ 
as the set of 
edges $e \in E$ such that 
$x(e) > 0$.
For each subset $F \subseteq E$ and 
each subset $X \subseteq V$, we define $F\langle X \rangle$ 
as the set of edges $e \in F$ such that $e \subseteq X$. 

For each vertex $v \in V$ and 
each non-empty subset $F \subseteq E(v)$, 
$F$ is said to be \emph{flat} 
if $e \sim_v f$ for every pair of edges $e,f \in F$. 
For each vertex $v \in V$, 
each symbol $\odot \in \{\succsim_v,\succ_v,\sim_v\}$, and 
each pair of flat subsets $F_1, F_2 \subseteq E(v)$, 
we 
write $F_1 \odot F_2$
if $e \odot f$ 
for every edge $e \in F_1$ and every edge $f \in F_2$. 
For each vertex $v \in V$,
each symbol $\odot \in \{\succsim_v,\succ_v,\sim_v\}$,  
each flat subset $F \subseteq E(v)$, 
and each edge $e \in E(v)$, 
we 
write $e \odot F$ (resp.\ $F \odot e$) 
if $e \odot f$ (resp.\ $f \odot e$) 
for every edge $f \in F$. 
In addition, for each vertex $v \in V$, 
each symbol $\odot \in \{\succsim_v,\succ_v,\sim_v\}$, and
each edge $e \in E(v)$, 
we define 
$E[\mathop{\odot} e]$ (resp.\ $E[e \mathop{\odot}]$) as the set of 
edges $f \in E(v)$
such that $f \mathop{\odot} e$
(resp.\ $e \mathop{\odot} f$). 

\section{Useful Lemmas}
\label{section:useful_lemmas}

Define ${\bf P}$ as the set of vectors $x \in \mathbb{R}_+^E$
satisfying the following conditions. 
\begin{equation} \label{eq_1:constraint} 
x(E(v)) \le 1 
 \ \ \ \mbox{($\forall v \in V$)}.
\end{equation}
\begin{equation} \label{eq_2:constraint} 
\displaystyle{x(E[\sim_v e]) + \sum_{w \in e}x(E[\succ_w e]) \ge 1} 
 \ \ \ \mbox{($\forall e \in E$, $\forall v \in e$)}. 
\end{equation}
\begin{equation} \label{eq_3:constraint} 
\displaystyle{x(E\langle X \rangle) \le \lfloor |X|/2 \rfloor} 
 \ \ \ \mbox{($\forall X \subseteq V$ such that $|X|$ is odd)}. 
\end{equation}
It should be noted that 
the set of vectors $x \in \mathbb{R}_+^E$
satisfying 
\eqref{eq_1:constraint}, \eqref{eq_2:constraint} 
coincides with the convex hull of 
the characteristic vectors of strongly stable matchings 
in the stable matching problem with ties in a bipartite graph~\cite{Kunysz18}.
Furthermore, 
\eqref{eq_3:constraint} 
is used to 
describe the convex hull of 
the characteristic vectors of matchings in 
a general graph~\cite{Edmonds65}. 

Our goal is to prove that 
we can check the existence of a strongly stable matching in $G$ by 
solving 
some linear programs over ${\bf P}$ under a constraint 
that the values for some edges have to be $0$. 
To this end, in this section, we prove useful properties of ${\bf P}$. 

\begin{lemma} \label{lemma:polytope} 
Let $\mu$ be a strongly 
stable matching in $G$.
Then $\chi_{\mu} \in {\bf P}$.
\end{lemma}
\begin{proof}
Since $\mu$ is a matching in $G$, 
\eqref{eq_1:constraint} and 
\eqref{eq_3:constraint} are 
clearly satisfied. 
Let $e$ and $v$ be an edge in $E$ and a vertex in $e$, 
respectively. 
Then we consider \eqref{eq_2:constraint} for $e$ and $v$. 
If $e \in \mu$, then $\chi_{\mu}(e) = 1$. 
Thus, 
\eqref{eq_2:constraint} holds. 
Assume that 
$e \notin \mu$. 
If $\mu(v) \succsim_v e$, then 
$\chi_{\mu}(E[\succsim_v e]) = 1$. 
Thus, 
\eqref{eq_2:constraint} holds. 
On the other hand, if $e \succ_v \mu(v)$, then 
since $\mu$ is strongly stable, 
$\mu(w) \succ_w e$, where 
we assume that $e = \{v,w\}$. 
Thus, 
\eqref{eq_2:constraint} holds. 
This completes the proof. 
\end{proof} 

Throughout this paper, 
we assume that 
${\bf P} \neq \emptyset$. 
Otherwise, Lemma~\ref{lemma:polytope} 
implies that 
there does not exist a strongly stable matching in $G$. 

\begin{lemma} \label{lemma:integral}
Let $x$ be an element in ${\bf P} \cap \{0,1\}^E$.
Define $\mu \coloneqq \{e \in E \mid x(e) = 1\}$.
Then $\mu$ is a strongly stable matching in $G$. 
\end{lemma}
\begin{proof}
Since \eqref{eq_1:constraint} implies that 
$\mu$ is a matching in $G$, 
we prove that 
$\mu$ is strongly stable. 
Assume that there exists an edge $e \in E \setminus \mu$ that 
blocks $\mu$. 
Then $e \succsim_w \mu(w)$ for every vertex $w \in e$, and 
there exists a vertex $v \in e$ such that 
$e \succ_v \mu(v)$. 
In this case, the left-hand side of \eqref{eq_2:constraint} 
for $e$ and $v$ is equal to $0$. 
This contradicts the fact that $x \in {\bf P}$. 
This completes the proof. 
\end{proof} 

The proof of the following lemma is basically the same as 
the proof of \cite[Lemma~12]{Kunysz18}. 
For readers' convenience, 
we give its proof 
since the proof is not contained in \cite{Kunysz18}. 

\begin{lemma} \label{lemma:self_dual} 
For every pair of vectors $x,z \in \mathbb{R}_+^E$ 
satisfying \eqref{eq_1:constraint}, \eqref{eq_2:constraint}  
and every edge $e \in E_z$, 
the following statements hold. 
\begin{description}
\item[(S1)]
$x(E(v)) = 1$ for every vertex $v \in e$. 
\item[(S2)]
$x(E[\sim_v e]) + \sum_{w \in e}x(E[\succ_w e]) = 1$ 
for every vertex $v \in e$. 
\item[(S3)] 
$x(E[\sim_v e]) 
=
x(E[\sim_w e])$, where 
we assume that $e = \{v,w\}$.
\end{description}
\end{lemma}
\begin{proof} 
We consider the following linear program.
\begin{equation} \label{eq_1:self_dual}
\mbox{Maximize} \ \ 
 x(E)  \vspace{2mm} \ \ \ \ 
\mbox{subject to} \ \ \ \ 
\eqref{eq_1:constraint}, \ 
\eqref{eq_2:constraint}, \ 
x \in \mathbb{R}_+^E.
\end{equation}
Define $\mathcal{T}$ as the set of ordered pairs 
$(e,v)$ of an edge $e \in E$ and a vertex $v \in e$. 
Then the dual problem of \eqref{eq_1:self_dual} is described as follows. 
\begin{equation} \label{eq_2:self_dual}
\begin{array}{cl}
\mbox{Minimize} \ \ 
&
\displaystyle{\alpha(V) 
-
\sum_{e \in E}\sum_{v \in e}\beta(e,v)
}\vspace{2mm}\\ 
\mbox{subject to} \ \ 
& 
\displaystyle{
\sum_{v \in e}
\Big(\alpha(v)
-
\sum_{f \in E[\sim_v e]}\beta(f,v)
- 
\sum_{f \in E[e \succ_v]}\sum_{w \in f}\beta(f,w)\Big)
\ge 1} \ \ 
\ \ \ \mbox{($\forall e \in E$)} \vspace{1mm}\\
& 
(\alpha,\beta) \in \mathbb{R}^V_+ \times \mathbb{R}_+^{\mathcal{T}}.
\end{array}
\end{equation}
For each feasible solution $x$ to \eqref{eq_1:self_dual}, 
we define $\alpha_x \in \mathbb{R}_+^V$ and 
$\beta_x \in \mathbb{R}_+^{\mathcal{T}}$
by 
\begin{equation*}
\alpha_x(v) \coloneqq 
x(E(v)), \ \ \ 
\beta_x(e,v) \coloneqq \dfrac{x(e)}{2}.
\end{equation*}

\begin{claim} \label{claim:lemma_1:lemma:self_dual} 
For every feasible solution $x$ to \eqref{eq_1:self_dual}, 
$(\alpha_x,\beta_x)$ is a feasible solution to 
\eqref{eq_2:self_dual}. 
\end{claim}
\begin{proof} 
Let $x$ be a feasible solution to \eqref{eq_1:self_dual}.
Then for every edge $e \in E$, 
\begin{equation*}
\begin{split}
&
\sum_{v \in e}
\Big(\alpha_x(v)
-
\sum_{f \in E[\sim_v e]}\beta_x(f,v)
- 
\sum_{f \in E[e \succ_v]}\sum_{w \in f}\beta_x(f,w)\Big)\\
&
=
\sum_{v \in e}
\Big(x(E(v))
-
\dfrac{x(E[\sim_v e])}{2}
- 
x(E[e \succ_v])\Big)
=
\sum_{v \in e}
\Big(x(E(v))
-
x(E[e \succsim_v])
+ 
\dfrac{x(E[\sim_v e])}{2}\Big)\\
&
=
\sum_{v \in e}
\Big(x(E[\succ_v e])
+ 
\dfrac{x(E[\sim_v e])}{2}\Big)
=
\dfrac{1}{2}
\sum_{v \in e}
\Big(x(E[\sim_v e])
+
\sum_{w \in e}x(E[\succ_w e])
\Big)
\ge
\dfrac{1}{2}
\cdot 2 = 1. 
\end{split} 
\end{equation*}
This completes the proof. 
\end{proof} 

\begin{claim} \label{claim:lemma_2:lemma:self_dual} 
For every feasible solution $x$ to \eqref{eq_1:self_dual}, 
$x$ and $(\alpha_x,\beta_x)$ are optimal solutions to 
\eqref{eq_1:self_dual} and 
\eqref{eq_2:self_dual}, respectively.
\end{claim}
\begin{proof} 
For every feasible solution $x$ to \eqref{eq_1:self_dual}, 
since 
\begin{equation*}
\alpha_x(V) 
-
\sum_{e \in E}\sum_{v \in e}\beta_x(e,v)
=
\sum_{v \in V}x(E(v))
- 
\sum_{e \in E}\sum_{v \in e}\dfrac{x(e)}{2}
=
2x(E)
- 
x(E) = x(E), 
\end{equation*}
Claim~\ref{claim:lemma_1:lemma:self_dual}
and 
the duality theorem of linear programming
(see, e.g., \cite[Theorem~3.7]{ConfortiCZ14}) imply 
that 
$x$ and $(\alpha_x,\beta_x)$ are optimal solutions to 
\eqref{eq_1:self_dual} and 
\eqref{eq_2:self_dual}, respectively.
\end{proof}

Let $x,z$ be 
vectors in $\mathbb{R}_+^E$ 
satisfying \eqref{eq_1:constraint}, \eqref{eq_2:constraint}. 
Then Claim~\ref{claim:lemma_2:lemma:self_dual}
implies that 
$x$ and $(\alpha_z,\beta_z)$ are optimal solutions to 
\eqref{eq_1:self_dual} and 
\eqref{eq_2:self_dual}, respectively.
Let 
$e$ be an edge in $E_z$.
Then we have 
$\alpha_z(v) > 0$ and
$\beta_z(e,v) > 0$
for every vertex $v \in e$.
Thus, the 
complementary slackness theorem of linear programming
(see, e.g., \cite[Theorem 3.8]{ConfortiCZ14}) 
implies (S1) for every vertex $v \in e$ and (S2)
for every vertex $v \in e$.  
Furthermore, (S2) implies that  
\begin{equation*}
x(E[\sim_v e]) + \sum_{u \in e}x(E[\succ_u e]) 
=
1
=
x(E[\sim_w e]) + \sum_{u \in e}x(E[\succ_u e]),
\end{equation*}
where we assume that $e = \{v,w\}$. 
This implies (S3).
This completes the proof. 
\end{proof} 

Define $V_0$ as the set of vertices $v \in V$ such that 
$x(E(v)) = 0$ 
for every element $x \in {\bf P}$. 
Then we define $V_1 \coloneqq V \setminus V_0$. 
That is, 
$V_1$ is the set of vertices $v \in V$ such that 
$x(E(v)) > 0$ for some element $x \in {\bf P}$. 
Notice that, for every element $x \in {\bf P}$, 
we have $E_x \subseteq E\langle V_1 \rangle$. 

\begin{lemma} \label{lemma:partition}
For every vertex $v \in V_1$ and 
every element $x \in {\bf P}$, 
we have $x(E(v)) = 1$.
\end{lemma}
\begin{proof}
Let $v$ be a vertex in $V_1$. 
The definition of $V_1$ implies that 
there exists an element $z \in {\bf P}$
such that 
$z(E(v)) > 0$. 
This implies that there exists an edge $e \in E(v)$ such that 
$e \in E_z$. 
Thus, (S1) of 
Lemma~\ref{lemma:self_dual}
implies that 
$x(E(v)) = 1$ for every element $x \in {\bf P}$. 
\end{proof} 

\section{Algorithm}
\label{section:algorithm} 

In this section, we explain our algorithm 
for the problem of checking the existence of a strongly 
stable matching in $G$. 

For each element $x \in {\bf P}$ and each 
vertex $v \in V_1$, we define 
$B_x(v)$ (resp.\ $W_x(v)$) as the set of 
edges $e \in E_x(v)$
such that 
$e \succsim_v f$ 
(resp.\ $f \succsim_v e$) 
for every edge $f \in E_x(v)$.
That is, 
$B_x(v)$ (resp.\ $W_x(v)$) 
is the set of the best (resp.\ worst) edges in $E_x(v)$ for $v$. 
Notice that 
$B_x(v)$ and $W_x(v)$ are flat. 
For each element 
$x \in {\bf P}$, we define 
$T_x$ as the set of 
edges $e \in E\langle V_1 \rangle$
satisfying one of the following conditions for each vertex $v \in e$. 
\begin{itemize}
\item
$e \in B_x(v) \cup W_x(v)$.
\item
$B_x(v) \succ_v e \succ_v W_x(v)$. 
\end{itemize}
Notice that, for every element $x \in {\bf P}$, 
$E_x \subseteq T_x$. 
For each subset $F \subseteq E$, we define 
${\bf P}(F)$ by 
\begin{equation*}
{\bf P}(F) \coloneqq
\{x \in {\bf P} \mid 
\mbox{$x(e) = 0$ for every edge 
$e \in E \setminus F$}\}. 
\end{equation*}

The following lemma plays an important role in 
the proof of the correctness of 
the proposed algorithm. 
We give the proof of Lemma~\ref{lemma:existence} in 
Section~\ref{section:existence}. 

\begin{lemma} \label{lemma:existence} 
Let $x$ be an element in ${\bf P}$. 
Assume that there exists a strongly stable matching 
in $G$. 
Then there exists a strongly stable matching $\mu$ in $G$
such that $\chi_{\mu} \in {\bf P}(T_x)$. 
\end{lemma}

The proposes algorithm is described in Algorithm~\ref{alg:main}. 
Recall that, in this paper, we assume that ${\bf P} \neq \emptyset$. 
In the course of the algorithm, since $z_t \in {\bf P}(T_{z_t})$, 
we have ${\bf P}(T_{z_t}) \neq \emptyset$. 

\begin{algorithm}[ht]
Find an element $z_1 \in {\bf P}$. 
Set $t \coloneqq 1$. \\
\While{there exists an edge $f_t \in E_{z_t}$ such that $z_t(f_t) \neq 1$}
{
  Let $v_t$ be a vertex in $f_t$, and let $g_t$ be an edge in $W_{z_t}(v_t)$.\\
  Find an element $a_t \in {\bf P}(T_{z_t})$ maximizing $a_t(g_t)$ 
  among all the elements in ${\bf P}(T_{z_t})$.\\
  \uIf{$a_t(g_t) = 1$}
  {
    Define $z_{t+1} \coloneqq a_t$. 
  }
  \Else
  {
    Find an element $b_t \in {\bf P}(T_{z_t})$ minimizing $b_t(g_t)$ 
    among all the elements in ${\bf P}(T_{z_t})$.\\
    \uIf{$b_t(g_t) = 0$}
    {
      Define $z_{t+1} \coloneqq b_t$. 
    }
    \Else
    {
      Output {\bf No} and halt. 
    }
    Set $t \coloneqq t + 1$. \\
  }
}
Output $z_t$ and halt. 
\caption{Algorithm for checking the existence of a strongly stable matching}
\label{alg:main}
\end{algorithm}

We first prove that the number of iterations in 
Algorithm~\ref{alg:main} is bounded by a polynomial in the input 
size.
In the course of Algorithm~\ref{alg:main}, 
since $z_{t+1} \in {\bf P}(T_{z_t})$, 
$E_{z_{t+1}} \subseteq T_{z_t}$.
Thus, for every vertex $v \in V_1$, 
$W_{z_{t+1}}(v) \succsim_v W_{z_{t}}(v)$ and 
$B_{z_{t}}(v) \succsim_v B_{z_{t+1}}(v)$. 

\begin{lemma} \label{lemma:alg_iteration_sub}
In the course of Algorithm~\ref{alg:main},
we have $T_{z_{t+1}} \subseteq T_{z_t}$. 
\end{lemma} 
\begin{proof}
Let $e$ be an edge in $T_{z_{t+1}}$. 
Then we prove that $e \in T_{z_t}$. 
Let us fix a vertex $v \in e$. 

If $B_{z_{t+1}}(v) \succ_v e \succ_v W_{z_{t+1}}(v)$, then  
\begin{equation*}
B_{z_t}(v) \succsim_v B_{z_{t+1}}(v) \succ_v e \succ_v W_{z_{t+1}}(v)
\succsim_v W_{z_t}(v). 
\end{equation*}
Thus, we can assume that 
$e \in B_{z_{t+1}}(v) \cup W_{z_{t+1}}(v)$. 

Assume that $e \in B_{z_{t+1}}(v) \cap W_{z_{t+1}}(v)$. 
If $B_{z_t}(v) \sim_v B_{z_{t+1}}(v)$,
then we have 
$e \in B_{z_{t+1}}(v) \subseteq B_{z_t}(v)$. 
Furthermore, 
if $W_{z_{t+1}}(v) \sim_v W_{z_t}(v)$, 
then 
we have 
$e \in W_{z_{t+1}}(v) \subseteq W_{z_t}(v)$. 
Thus, we can assume that 
$B_{z_t}(v) \succ_v B_{z_{t+1}}(v)$
and 
$W_{z_{t+1}}(v) \succ_v W_{z_t}(v)$.
Then 
\begin{equation*}
B_{z_t}(v) \succ_v B_{z_{t+1}}(v) \succsim_v e \succsim_v W_{z_{t+1}}(v)
\succ_v W_{z_t}(v). 
\end{equation*}

Assume that $e \in W_{z_{t+1}}(v) \setminus B_{z_{t+1}}(v)$. 
Since $e \notin B_{z_{t+1}}(v)$, 
$B_{z_t}(v) \succsim_v B_{z_{t+1}}(v) \succ_v e$. 
Thus, if $W_{z_{t+1}}(v) \succ_v W_{z_t}(v)$, then 
\begin{equation*}
B_{z_t}(v) \succ_v e \succsim_v W_{z_{t+1}}(v) \succ_v W_{z_t}(v).
\end{equation*}
This implies that 
we can 
assume that 
$W_{z_{t+1}}(v) \sim_v W_{z_t}(v)$.
Then 
$e \in W_{z_{t+1}}(v) \subseteq W_{z_t}(v)$.

Assume that $e \in B_{z_{t+1}}(v) \setminus W_{z_{t+1}}(v)$. 
Since $e \notin W_{z_{t+1}}(v)$, 
$e \succ_v W_{z_{t+1}}(v) \succsim_v W_{z_t}(v)$. 
Thus, if $B_{z_t}(v) \succ_v B_{z_{t+1}}(v)$, then 
\begin{equation*}
B_{z_{t}}(v) \succ_v B_{z_{t+1}}(v) \succsim_v e \succ_v W_{z_{t}}(v).
\end{equation*}
This implies that 
we can 
assume that 
$B_{z_t}(v) \sim_v B_{z_{t+1}}(v)$.
Then 
$e \in B_{z_{t+1}}(v) \subseteq B_{z_t}(v)$.
\end{proof} 

\begin{lemma} \label{lemma:alg_iteration}
In the course of Algorithm~\ref{alg:main},
we have $T_{z_{t+1}} \neq T_{z_t}$. 
\end{lemma} 
\begin{proof}
We first consider the case where $z_{t+1} = a_t$. 
Since 
$v_t \in V_1$ and 
$z_t(f_t) \neq 1$, Lemma~\ref{lemma:partition} 
implies that 
there exists an edge $e \in B_{z_t}(v_t)$
such that $e \neq g_t$. 
Since $a_{t}(g_t) = 1$, 
we have $a_{t}(e) = 0$ (i.e., 
$e \notin E_{z_{t+1}}$). 
Thus, if $e \in T_{z_{t+1}}$, then 
$B_{z_{t+1}}(v_t) \succ_{v_t} e \sim_{v_t} B_{z_{t}}(v_t)$, which 
contradicts the fact that 
$B_{z_{t}}(v_t) \succsim_{v_t} B_{z_{t+1}}(v_t)$. 
This implies that $e \notin T_{z_{t+1}}$.
Since $e \in B_{z_t}(v_t) \subseteq T_{z_t}$, 
$T_{z_{t+1}} \neq T_{z_t}$. 

Next, we consider the case where 
$z_{t+1} = b_t$. 
Since $b_t(g_t) = 0$ (i.e., $g_t \notin E_{z_{t+1}}$), 
if $g_t \in T_{z+1}$, then 
$W_{z_t}(v_t) \sim_{v_t} g_t \succ_{v_t} W_{z_{t+1}}(v_t)$, which 
contradicts the fact that 
$W_{z_{t+1}}(v_t) \succsim_{v_t} W_{z_{t}}(v_t)$. 
Thus, 
$g_t \notin T_{z+1}$. 
Since $g_t \in W_{z_t}(v_t) \subseteq T_{z_t}$, 
$T_{z_{t+1}} \neq T_{z_t}$. 
\end{proof} 

Lemmas~\ref{lemma:alg_iteration_sub} and 
\ref{lemma:alg_iteration} imply that the number of 
iterations in Algorithm~\ref{alg:main} is $O(|E|)$. 
We can find $a_t,b_t$ by solving linear programs over
${\bf P}(z_t)$.
In addition, since it is known that 
the separation problem for \eqref{eq_3:constraint} can be solved 
in polynomial time~\cite{PadbergR82}, 
a linear program over 
${\bf P}(z_t)$ can be solved in polynomial time~\cite[Theorem~6.4.9]{GrotschelLS93}. 
Thus, Algorithm~\ref{alg:main} is a polynomial-time algorithm.  

\begin{theorem}
Algorithm~1 can correctly solve the problem of checking the existence of 
a strongly stable matching in $G$. 
\end{theorem}
\begin{proof}
Assume that 
Algorithm~\ref{alg:main} halts when $t = k$. 

First, we assume that Algorithm~\ref{alg:main} outputs 
$z_k$ in Step~17. 
Then $z_k \in \{0,1\}^E$. 
In addition, $z_k \in {\bf P}$.
Thus, by defining 
$\mu \coloneqq \{e \in E \mid z_k(e) = 1\}$, 
Lemma~\ref{lemma:integral}
implies that $\mu$ is a strongly stable matching in $G$.  

Next, we consider the case where 
Algorithm~\ref{alg:main} outputs 
{\bf No}. 
In this case, 
$0 < x(g_k) < 1$ holds
for every element $x \in {\bf P}(T_{z_k})$.
In order to derive a contradiction, 
we assume that there exists a strongly stable matching in $G$. 
Lemma~\ref{lemma:existence} 
implies that 
there exists a strongly stable matching $\mu$ in 
$G$ such that $\chi_{\mu} \in {\bf P}(T_{z_k})$. 
Since $\chi_{\mu}(g_k) \in \{0,1\}$, 
this is a contradiction. 
\end{proof} 

\subsection{Proof of Lemma~\ref{lemma:existence}}
\label{section:existence} 

In this subsection, 
we prove Lemma~\ref{lemma:existence}.

Throughout this subsection, 
let $\sigma$ be a strongly stable matching in $G$. 
Define $P$ (resp.\ $Q$) as the set of vertices $v \in V_1$ satisfying one of the 
following conditions. 
\begin{itemize}
\item
$W_x(v) \succ_v \sigma(v)$ (resp.\ $\sigma(v) \succ_v B_x(v)$). 
\item
$B_x(v) \succ_v W_x(v)$ and 
$W_x(v) \sim_v \sigma(v)$
(resp.\ $B_x(v) \succ_v W_x(v)$ and 
$B_x(v) \sim_v \sigma(v)$). 
\end{itemize}
Define $R$ as the set of vertices $v \in V_1$ such that 
$E_x(v)$ is flat and 
$E_x(v) \sim_v \sigma(v)$. 
Notice that the subsets 
$P,Q,R \subseteq V_1$ are pairwise disjoint.
Define 
$S \coloneqq V_1 \setminus (P \cup Q \cup R)$. 
Then for every vertex $v \in S$, 
$B_x(v) \succ_v \sigma(v) \succ_v W_x(v)$. 

Define $E^+$ by 
\begin{equation*}
E^+ \coloneqq \{e \in E \mid \exists z \in {\bf P} \colon e \in E_z\}
\subseteq E\langle V_1 \rangle. 
\end{equation*}
Notice that, for  
every vertex $v \in V_1$, 
since $W_x(v), B_x(v) \subseteq E_x$, 
we have 
$W_x(v), B_x(v) \subseteq E^+$. 
In addition,  
since $\chi_{\sigma} \in {\bf P}$
and $\chi_{\sigma}(e) > 0$ 
for every edge $e \in \sigma$, 
we have 
$\sigma \subseteq E^+ $. 

\begin{lemma} \label{lemma:sim}
Let $e = \{v,w\}$ be an edge in $E^+$. 
Then we have the following statements. 
\begin{description}
\item[(A1)]
$W_x(v) \succ_v e$ if and only if 
$e \succ_w B_x(w)$. 
\item[(A2)]
$W_x(v) \sim_v e$ if and only if 
$B_x(w) \sim_w e$. 
\item[(A3)] 
$E_x(v)$ is flat and $E_x(v) \sim_v e$ 
if and only if $E_x(w)$ is flat
and $E_x(w) \sim_w e$. 
\item[(A4)]
$e \in W_x(v)$ if and only if $e \in B_x(w)$. 
\end{description}
\end{lemma} 
\begin{proof}
Since $e \in E^+$, 
(S2) of Lemma~\ref{lemma:self_dual} implies that 
\begin{equation} \label{eq_1:lemma:sim} 
x(E[\succsim_v e]) + x(E[\succ_w e]) = 1, \ \ 
x(E[\succ_v e]) + x(E[\succsim_w e]) = 1. 
\end{equation}

{\bf (A1)} 
The second equation of \eqref{eq_1:lemma:sim} 
implies that 
$x(E[\succ_v e]) = 1$ if and only if
$x(E[\succsim_w e]) = 0$. 
This implies that 
$W_x(v) \succ_v e$ if and only if 
$e \succ_w B_x(w)$. 

{\bf (A2)} 
The first equation of 
\eqref{eq_1:lemma:sim} 
implies that 
$x(E[\succsim_v e]) = 1$ if and only if
$x(E[\succ_w e]) = 0$. 
This implies that 
$W_x(v) \succsim_v e$ if and only if 
$e \succsim_w B_x(w)$. 
Thus, (A1) implies this statement.

{\bf (A3)}
Assume that 
$E_x(v)$ is flat and $E_x(v) \sim_v e$. 
Then since 
$W_x(v) \sim_v e$ and 
$B_x(v) \sim_v e$, 
(A2) implies that 
$W_x(w) \sim_w e$ and 
$B_x(w) \sim_w e$.
This implies that $E_x(w)$ is flat. 
Furthermore, since 
$e \sim_v E_x(v) \sim_v W_x(v)$, 
(A2) implies that 
$e \sim_w B_x(w) \sim_w E_x(w)$.  

{\bf (A4)} 
Since $e \in E_x(v)$ if and only if 
$e \in E_x(w)$. 
Thus, (A2) implies this statement.
\end{proof} 

\begin{lemma} \label{lemma:sigma}
Let $e = \{v,w\}$ be an edge in $\sigma$. 
Then we have the following statements. 
\begin{description}
\item[(B1)]
$v \in P$ if and only if $w \in Q$.
\item[(B2)]
$e \subseteq R$ or $e \cap R = \emptyset$. 
\end{description} 
\end{lemma}
\begin{proof}
Recall that since $\chi_{\sigma} \in {\bf P}$, we have 
$e \in E^+$. 

{\bf (B1)}
Assume that $v \in P$. 
If $W_x(v) \succ_v e$, then 
(A1) of Lemma~\ref{lemma:sim}
implies that 
$e \succ_w B_x(w)$.
Thus, $w \in Q$. 
Assume that 
$B_x(v) \succ_v W_x(v)$ and 
$W_x(v) \sim_v e$. 
Then 
(A2) of Lemma~\ref{lemma:sim}
implies that 
$B_x(w) \sim_w e$.
Thus, if 
$E_x(w)$ is flat, then 
(A3) of Lemma~\ref{lemma:sim}
implies 
that $E_x(v)$ is flat. 
This contradicts the fact that 
$B_x(v) \succ_v W_x(v)$. 
Thus, 
$B_x(w) \succ_w W_x(w)$. 
This implies that $w \in Q$. 

In addition, the opposite 
direction 
can be proved in the same way. 

{\bf (B2)}
This statement follows from (A3) of Lemma~\ref{lemma:sim}. 
\end{proof}

Lemma~\ref{lemma:sigma} implies that, 
for each edge $e \in \sigma$, 
we have exactly one of 
(i) 
$|e \cap P| = |e \cap Q| = 1$, 
(ii)
$e \subseteq R$,
(iii)
$e \subseteq S$.

\begin{lemma} \label{lemma:bipartite_r} 
Let $v$ be a vertex in $R$, and let $e = \{v,w\}$ be 
an edge in $E_x(v)$. 
Then $w \in R$. 
\end{lemma}
\begin{proof}
Since $v \in R$, 
$E_x(v)$ is flat. 
Since $e \in E_x(v)$, 
$E_x(v) \sim_v e$. 
Thus, since $e \in E_x \subseteq E^+$, 
(A3) of Lemma~\ref{lemma:sim} 
implies that 
$E_x(w)$ is flat. 
Since $\chi_{\sigma} \in {\bf P}$ and $e \in E_x$, 
(S2) of Lemma~\ref{lemma:self_dual}
implies that 
\begin{equation*}
\chi_{\sigma}(E[\succsim_v e]) + \chi_{\sigma}(E[\succ_w e]) = 1, \ \ \ 
\chi_{\sigma}(E[\succ_v e]) + \chi_{\sigma}(E[\succsim_w e]) = 1. 
\end{equation*}
Since $e \sim_v \sigma(v)$ follows from 
$v \in R$, 
we have $\chi_{\sigma}(E[\succsim_v e]) = 1$. 
This implies that 
$\chi_{\sigma}(E[\succ_w e]) = 0$. 
Thus, $e \succsim_w \sigma(w)$. 
Furthermore, since $\chi_{\sigma}(E[\succ_v e]) = 0$, 
$\chi_{\sigma}(E[\succsim_w e]) = 1$.
This implies that 
$\sigma(w) \succsim_w e$. 
Thus, $e \sim_w \sigma(w)$. 
Since $e \in E_x(w)$, $E_x(w) \sim_w \sigma(w)$.
Thus, $w \in R$. 
\end{proof} 

\begin{lemma} \label{lemma:bipartite_p} 
Let $v$ be a vertex in $P$, and let $e = \{v,w\}$ be 
an edge in $W_x(v)$. 
Then $w \in Q$. 
\end{lemma}
\begin{proof}
Since $e \in E^+$ and $e \in W_x(v)$, 
(A4) of 
Lemma~\ref{lemma:sim} implies that 
$e \in B_x(w)$. 
Furthermore, 
since $\chi_{\sigma} \in {\bf P}$ and $e \in E_x$, 
(S2) of Lemma~\ref{lemma:self_dual}
implies that 
\begin{equation*}
\chi_{\sigma}(E[\succsim_v e]) + \chi_{\sigma}(E[\succ_w e]) = 1, \ \ \ 
\chi_{\sigma}(E[\succ_v e]) + \chi_{\sigma}(E[\succsim_w e]) = 1. 
\end{equation*}

Assume that $W_x(v) \succ_v \sigma(v)$.
Then since $e \in W_x(v)$, 
$\chi_{\sigma}(E[\succsim_v e]) = 0$. 
Thus, $\chi_{\sigma}(E[\succ_w e]) = 1$. 
Since $e \in B_x(w)$, 
this implies that $\sigma(w) \succ_w B_x(w)$. 

Assume that 
$B_x(v) \succ_v W_x(v)$ and 
$W_x(v) \sim_v \sigma(v)$. 
Since $e \in B_x(w)$, 
$B_x(w) \sim_w e$. 
Thus, 
if $E_x(w)$ is flat, then 
since $e \in W_x(v) \subseteq E^+$, 
(A3) of 
Lemma~\ref{lemma:sim} implies that 
$E_x(v)$ is flat.
However, this contradicts the fact that 
$B_x(v) \succ_v W_x(v)$.
Thus, $B_x(w) \succ_w W_x(w)$.
Since $e \sim_v \sigma(v)$ follows from 
$e \in W_x(v)$, 
$\chi_{\sigma}(E[\succsim_v e]) = 1$. 
Thus, $\chi_{\sigma}(E[\succ_w e]) = 0$. 
This implies that 
$e \succsim_w \sigma(w)$. 
Since 
$\chi_{\sigma}(E[\succ_v e]) = 0$, 
$\chi_{\sigma}(E[\succsim_w e]) = 1$.
This implies that  
$\sigma(w) \succsim_w e$.
Thus, since $e \succsim_w \sigma(w)$, 
we have $\sigma(w) \sim_w e$.
Since $e \in B_x(w)$, 
$B_x(w) \sim_w \sigma(w)$.
Thus, $w \in Q$. 
\end{proof}

\begin{lemma} \label{lemma:bipartite_q} 
Let $v$ be a vertex in $Q$, and let $e = \{v,w\}$ be 
an edge in $B_x(v)$. 
Then $w \in P$. 
\end{lemma}
\begin{proof}
Since $e \in E^+$ and $e \in B_x(v)$, 
(A4) of 
Lemma~\ref{lemma:sim} implies that 
$e \in W_x(w)$. 
Furthermore, 
(S2) of Lemma~\ref{lemma:self_dual}
implies that 
\begin{equation*}
\chi_{\sigma}(E[\succsim_v e]) + \chi_{\sigma}(E[\succ_w e]) = 1, \ \ \ 
\chi_{\sigma}(E[\succ_v e]) + \chi_{\sigma}(E[\succsim_w e]) = 1. 
\end{equation*}

Assume that $\sigma(v) \succ_v B_x(v)$.
Then since $e \in B_x(v)$, 
$\chi_{\sigma}(E[\succ_v e]) = 1$. 
Thus, $\chi_{\sigma}(E[\succsim_w e]) = 0$. 
Since $e \in W_x(w)$, 
this implies that $W_x(w) \succ_w \sigma(w)$. 

Assume that 
$B_x(v) \succ_v W_x(v)$ and 
$B_x(v) \sim_v \sigma(v)$. 
Since $e \in W_x(w)$, 
$W_x(w) \sim_w e$. 
Thus, 
if $E_x(w)$ is flat, then 
since $e \in B_x(v) \subseteq E^+$, 
(A3) of 
Lemma~\ref{lemma:sim} implies that 
$E_x(v)$ is flat.
However, this contradicts the fact that 
$B_x(v) \succ_v W_x(v)$.
Thus, $B_x(w) \succ_w W_x(w)$.
Since $e \sim_v \sigma(v)$ follows from 
$e \in B_x(v)$, 
$\chi_{\sigma}(E[\succsim_v e]) = 1$. 
Thus, $\chi_{\sigma}(E[\succ_w e]) = 0$. 
This implies that 
$e \succsim_w \sigma(w)$. 
Since 
$\chi_{\sigma}(E[\succ_v e]) = 0$, 
$\chi_{\sigma}(E[\succsim_w e]) = 1$.
This implies that  
$\sigma(w) \succsim_w e$.
Thus, since $e \succsim_w \sigma(w)$, 
we have $\sigma(w) \sim_w e$.
Since $e \in W_x(w)$, 
$W_x(w) \sim_w \sigma(w)$.
Thus, $w \in P$. 
\end{proof}

Define $L \coloneqq \bigcup_{v \in P}W_x(v)$. 
Define the subgraph $G_L$ of $G$ by 
$G_L \coloneqq  (P \cup Q, L)$. 
Notice that Lemma~\ref{lemma:bipartite_p} 
implies that
$G_L$ is well-defined and  
a bipartite graph. 
In addition, 
(A4) of Lemma~\ref{lemma:sim} and 
Lemma~\ref{lemma:bipartite_q} imply that 
$L = \bigcup_{v \in Q}B_x(v)$. 
Define the vector $z_L \in \mathbb{R}_+^{L}$  
by 
\begin{equation*}
z_L(e) \coloneqq \dfrac{x(e)}{x(W_x(v))}, 
\end{equation*}
where we assume that $e \cap P \coloneqq \{v\}$. 

\begin{lemma} \label{lemma:l_1_equal}
Let $e = \{v,w\}$ be an edge in $L$. 
Assume that $w \in Q$. 
Then 
\begin{equation*}
z_L(e) = \dfrac{x(e)}{x(B_x(w))}.
\end{equation*}
\end{lemma} 
\begin{proof}

Since $x \in {\bf P}$, 
(S3) of Lemma~\ref{lemma:self_dual}
implies that 
$x(E[\sim_v e]) = x(E[\sim_w e])$. 
Since $e \in W_x(v)$, 
$x(E[\sim_v e]) = x(W_x(v))$. 
Since $e \in B_x(w)$, 
$x(E[\sim_w e]) = x(B_x(w))$. 
Thus, we have 
$x(W_x(v)) = x(B_x(w))$. 
This completes the proof. 
\end{proof} 

Define ${\bf A}$ as the set of vectors $y \in \mathbb{R}_+^L$ such that 
$y(L(v)) = 1$ for every vertex $v \in P \cup Q$. 
Since 
Lemma~\ref{lemma:l_1_equal} implies that
$z_L(L(v)) = 1$ for every 
vertex $v \in P \cup Q$, 
we have $z_L \in {\bf A}$. 
Since 
$G_L$ is bipartite, it is known that 
an extreme point of the 
polytope ${\bf A}$ 
coincides with
a restriction of $\chi_{\mu}$ to $L$ 
for some matching $\mu$ in $L$ that covers $P \cup Q$~\cite{Birkhoff46}
(see also \cite[Corollary~4.19]{ConfortiCZ14}). 
This implies that since $z_L \in {\bf A}$ (i.e., ${\bf A} \neq \emptyset$), 
there exists a matching $\mu_L$ in $L$ that covers 
$P \cup Q$.

Define $K \coloneqq \bigcup_{v \in R}E_x(v)$.  
Define the subgraph $G_K$ of $G$ by 
$G_K \coloneqq (R,K)$. 
Lemma~\ref{lemma:bipartite_r} 
implies that
$G_K$ is well-defined.
Define the vector $z_K \in \mathbb{R}_+^{K}$  
by $z_K(e) \coloneqq x(e)$. 
Then Lemma~\ref{lemma:partition} implies that 
$z_K(K(v)) = 1$ for every vertex $v \in R$. 
Define ${\bf B}$ as the set of 
vectors $y \in \mathbb{R}_+^{K}$ satisfying the 
following conditions. 
\begin{equation*} 
\begin{array}{cl}
& y(K(v)) \le 1 
\ \ \ \mbox{($\forall v \in R$)} \vspace{1mm} \\ 
&
\displaystyle{y(K\langle X \rangle) \le \lfloor |X|/2 \rfloor} 
\ \ \ \mbox{($\forall X \subseteq R$ such that $|X|$ is odd)}.
\end{array}
\end{equation*}
Then we consider the following 
linear problem. 
\begin{equation} \label{eq:matching_polytope} 
\mbox{Maximize} \ \  
y(K) \ \ \ \ 
\mbox{subject to} \ \ \ \ 
y \in {\bf B}. 
\end{equation}
Then since $x \in {\bf P}$, 
$z_K$ is a feasible solution to \eqref{eq:matching_polytope}.
Furthermore, $z_K(K) = |R|/2$.
Thus, since 
an objective value of \eqref{eq:matching_polytope} 
is at most $|R|/2$,  
$z_K$ is an optimal solution to \eqref{eq:matching_polytope}.  
It is known that 
an extreme point of the polytope ${\bf B}$ 
coincides with a restriction of $\chi_{\mu}$ 
to $K$ for some matching $\mu$ in $K$~\cite{Edmonds65} 
(see also \cite[Theorem~4.24]{ConfortiCZ14}). 
Furthermore, 
since ${\bf B}$ is a polytope, 
it is known 
that 
there exists an optimal solution to \eqref{eq:matching_polytope} that 
is 
an extreme point of ${\bf B}$ 
(see, e.g., \cite[Theorem~6.5.7]{GrotschelLS93}).
This implies that since 
$z_K(K) = |R|/2$, 
there exists a matching $\mu_K$ in $K$ 
that covers $R$. 

Define $\overline{\sigma}$ as the set of edges $e \in \sigma$ such that 
$e \subseteq S$. 
Define $\mu \coloneqq \mu_L \cup \mu_K \cup \overline{\sigma}$.  
Then $\mu$ is a matching in $G$. 
Since  
$E_x \subseteq T_x$
for every element $x \in {\bf P}$,
$\mu \subseteq T_x$, i.e., 
$\chi_{\mu}(e) = 0$ for every edge 
$e \in E \setminus T_x$. 
Thus, Lemma~\ref{lemma:polytope} implies that 
what remains is to prove that $\mu$ is strongly stable. 

\begin{lemma} \label{lemma:new_matching} 
$\mu$ is a strongly stable matching in $G$. 
\end{lemma}
\begin{proof}
Let $e = \{v,w\}$ be an edge in $E \setminus \mu$.
Then we prove that $e$ does not block $\mu$. 
We divide the proof into the following two cases.
\begin{description}
\item[Case~1.]
$e \in E \setminus (\mu \cup \sigma)$.
\item[Case~2.]
$e \in \sigma \setminus \mu$.
\end{description} 

{\bf Case~1.} 
The definition of $\mu$ implies that, for every vertex $u \in e$, 
$\mu(u) = \emptyset$ if and only if 
$\sigma(u) = \emptyset$. 
Thus, if 
$\mu(u) = \emptyset$ for every vertex $u \in e$, 
then 
$\sigma(u) = \emptyset$ for every vertex $u \in e$.
This implies that if $\mu(u) = \emptyset$ for every vertex $u \in e$, then 
$e$ blocks $\sigma$. 
Since $\sigma$ is strongly stable, 
there exists a vertex $u \in e$ such that 
$\mu(u) \neq \emptyset$. 
Without loss of generality, 
we assume that 
$\mu(v) \neq \emptyset$.

Since 
Lemma~\ref{lemma:polytope} implies that 
$\chi_{\sigma} \in {\bf P}$,  
if $\mu(w) = \emptyset$, then
$w \in V_0$.
Since $x \in {\bf P}$,  
\begin{equation*}
x(E[\succsim_v e]) + x(E[\succ_w e]) \ge 1, \ \ \ 
x(E[\succ_v e]) + x(E[\succsim_w e]) \ge 1. 
\end{equation*}
We divide the proof of this case into the following cases. 
\begin{description}
\item[(1-a)]
$v \in Q$ and
$\mu(w) = \emptyset$. 
\item[(1-b)]
$v \in V_1 \setminus Q$ and 
$\mu(w) = \emptyset$. 
\item[(1-c)]
$v \in Q$ and $w \in V_1$. 
\item[(1-d)]
$v \in V_1 \setminus Q$ and $w \in V_1 \setminus Q$. 
\end{description} 

{\bf (1-a)}
Since 
$x(E[\succ_v e]) \ge 1$ follows from $w \in V_0$, 
we have 
$B_x(v) \succ_v e$. 
Thus, since $\mu(v) \in B_x(v)$, 
we have $\mu(v) \succ_v e$.
This implies that $e$ does not block $\mu$. 

{\bf (1-b)}
Since $\sigma$ is strongly stable, 
$\sigma(v) \succ_v e$. 
Thus, since 
$\mu(v) \succsim_v \sigma(v)$, 
$e$ does not block $\mu$. 

{\bf (1-c)}
Assume that 
$e \succ_v B_x(v)$.
Then $x(E[\succsim_v e]) = 0$. 
Thus, $x(E[\succ_w e]) \ge 1$. 
This implies that
$W_x(w) \succ_w e$. 
Thus, 
since $\mu(w) \succsim_w W_x(w)$, 
$e$ does not block $\mu$. 

Assume that 
$B_x(v) \sim_v e$.
Then since 
$x(E[\succ_v e]) = 0$, 
$x(E[\succsim_w e]) \ge 1$. 
Thus, 
$W_x(w) \succsim_w e$. 
Since 
$\mu(v) \in B_x(v)$ and $\mu(w) \succsim_w W_x(w)$, 
$e$ does not block $\mu$. 

Assume that $B_x(v) \succ_v e$.
Since $\mu(v) \in B_x(v)$, 
$\mu(v) \succ_v e$. 
Thus, 
$e$ does not block $\mu$. 

{\bf (1-d)}
In this case, 
$\mu(v) \succsim_v \sigma(v)$ and 
$\mu(w) \succsim_w \sigma(w)$. 
Thus, since 
$\sigma$ is strongly stable, 
$e$ does not block $\mu$. 

{\bf Case~2.} 
In this case, 
we have exactly one of the following 
conditions. 
\begin{description}
\item[(2-a)]
$|e \cap P| = |e \cap Q| = 1$.
\item[(2-b)]
$e \subseteq R$. 
\end{description}

{\bf (2-a)}
Assume that 
$e \cap P = \{v\}$ and $e \cap Q = \{w\}$. 

Assume that 
$e \succ_w B_x(w)$. 
Then since $e \in E^+$, 
(A1) of Lemma~\ref{lemma:sim} implies that 
$W_x(v) \succ_v e$.
Since $\mu(v) \in W_x(v)$, 
$\mu(v) \succ_v e$. 
Thus, 
$e$ does not block $\mu$. 

Assume that 
$B_x(w) \sim_w e$. 
Then since $e \in E^+$, 
(A1) of Lemma~\ref{lemma:sim} implies that 
$W_x(v) \sim_v e$. 
Thus, since 
$\mu(v) \in W_x(v)$ and $\mu(w) \in B_x(w)$,  
$e$ does not block $\mu$. 

Assume that 
$B_x(w) \succ_w e$. 
Since $\mu(w) \in B_x(w)$, 
$\mu(w) \succ_w e$. 
Thus, 
$e$ does not block $\mu$. 

{\bf (2-b)}
In this case, $\mu(v) \sim_v e$ for every 
vertex $v \in e$. 
Thus, $e$ does not block $\mu$. 
\end{proof}

Lemma~\ref{lemma:new_matching} completes 
the proof of Lemma~\ref{lemma:existence}. 

\bibliographystyle{plain}
\bibliography{strong_roommates_lp_bib}

\end{document}